\documentclass[conference]{IEEEtran}

\newif\ifappendix
\appendixtrue
\ifappendix
    \newcommand{\appendixref}[2][]%
        {\ref{#2}}
\else
    \newcommand{\appendixref}[2][]%
        {\cite{Appendix}#1}
\fi

\IEEEoverridecommandlockouts
\usepackage{cite}
\usepackage{amsmath,amssymb,amsfonts}
\usepackage{xfrac}
\usepackage{algorithmic}
\usepackage{graphicx}
\usepackage{textcomp}
\usepackage{xcolor}
\usepackage{url}
\def\BibTeX{{\rm B\kern-.05em{\sc i\kern-.025em b}\kern-.08em
    T\kern-.1667em\lower.7ex\hbox{E}\kern-.125emX}}

\usepackage{amsthm}
\newtheorem{lemma}{lemma}
\newtheorem{claim}{claim}
\usepackage{comment}
\usepackage{hyperref}
\usepackage{orcidlink}
\usepackage{thmtools,thm-restate}

\begin{document}

\title{On Quorum Sizes in DAG-Based BFT Protocols
}

\author{
    \IEEEauthorblockN{Razya Ladelsky}
    \IEEEauthorblockA{
        Technion\\
        \texttt{ladelskyr@cs.technion.ac.il}
    }
    \and
    \IEEEauthorblockN{Roy Friedman}
    \IEEEauthorblockA{
        Technion\\
        \texttt{roy@cs.technion.ac.il}
    }
}

\maketitle
\thispagestyle{plain}
\pagestyle{plain}
\begin{abstract}
Several prominent DAG-based blockchain protocols, such as DAG-Rider, Tusk, and Bullshark, completely separate between equivocation elimination and committing; equivocation is handled through the use of a reliable Byzantine broadcast black-box protocol, while committing is handled by an independent DAG-based protocol.
With such an architecture, a natural question that we study in this paper is whether the DAG protocol would work when the number of nodes (or validators) is only $2f+1$ (when equivocation is eliminated), and whether there are benefits in working with larger number of nodes, i.e., a total of $kf+1$ nodes for $k > 3$.

We find that while DAG-Rider's correctness is maintained with $2f+1$ nodes, the asynchronous versions of both Tusk and Bullshark inherently depends on having $3f+1$ nodes, regardless of equivocation.
We also explore the impact of having larger number of nodes on the expected termination time of these three~protocols.
\end{abstract}


\section{Introduction}
\label{sec:intro}
Many Byzantine fault tolerant (BFT) consensus, state machine replication (SMR), and blockchain protocols are structured in a manner that repeatedly first eliminates potential equivocation of proposed values, and only then tries to agree on whether to decide on such a value~\cite{Tendermint,PBFT,Hotstuff}.
This is even more profound in several prominent DAG-based blockchain protocols, such as DAG-Rider~\cite{Dag-rider}, Tusk~\cite{Tusk}, and Bullshark~\cite{Bullshark}.
In fact, in these DAG-based protocols, values (or vertices in the DAG) are broadcast using a Byzantine reliable broadcast sub-protocol, which is treated as a black-box.
The received values (vertices) are then independently used by each node to form a local DAG structure, based on which each node attempts to decide on an increasingly longer chain of values (vertices).
Decided vertices are appended to the local chain, and may never be modified or removed.
Further, the local chains of all correct validators (nodes) must have a continuously growing common prefix.

All DAG protocols we are aware of target the optimal resiliency threshold of $n=3f+1$. Further, even their DAG construction, maintenance, and commit rules are based on having up to $3f+1$ vertices in each round of the DAG, with quorums of at least $2f+1$ vertices.
Yet, the clear separation between equivocation elimination and DAG maintenance in DAG-Rider~\cite{Dag-rider}, Tusk~\cite{Tusk}, and Bullshark~\cite{Bullshark} raises the question of whether this is really necessary.
That is, why not adopt DAGs of at most $2f+1$ vertices per round, with commit rules based on quorums of $f+1$ vertices?
After all, working with a smaller DAG implies less space, lower computational overheads, and lower communication overhead.

Alternatively, we may ask what is the impact of the overall ratio between $n$ and $f$ on other aspects of the protocol, and in particular on its complexity and expected latency.
For example, it is well known that for non-DAG based protocols, increasing the ratio between $n$ and $f$ may result in simpler and more efficient protocols~\cite{FMR05}, and may lead to faster expected termination~\cite{Zyzzyva,KTZ21,Fab,Bosco}.
Also, we point to the recent debate in the community regarding the best quorum sizes to be used in non-DAG based protocols when equipped with judicious use of a trusted execution environment (TEE)~\cite{Letittee}.
Such a TEE provides non-equivocation, but still allows Byzantine validators to try and interfere with the protocol's working in other ways~\cite{dissecting-bft}.
Arguments in favor and against using $n=2f+1$ vs.
$n=3f+1$ have been explored in~\cite{vivisectingdissection,dissecting-bft}.

The specific case of DAG-Rider has been recently partially explored in~\cite{Letittee}.
Specifically, it was shown in~\cite{Letittee} that DAG-Rider can be trivially rewritten to work with $n=2f+1$ and have its commit rule changed to be based on only $f+1$ vertices and the resulting algorithm, nicknamed TEE-rider, maintains liveness and safety.
However, the work in~\cite{Letittee} did not analyze the impact on the expected termination time.

In this work, we take a step towards filling this gap, by analyzing the impact of varying the number of validators participating in the DAG from $n=2f+1$, $n=3f+1$, and the general case of $n=kf+1, k\geq 3$ on DAG-Rider~\cite{Dag-rider}, Tusk~\cite{Tusk}, and Bullshark~\cite{Bullshark}, assuming equivocation has been eliminated separately.
We explore the impact on correctness and expected termination for these three protocols, and draw some general insights and~observations.

As a side effect, we also show a very simple TEE-less equivocation elimination mechanism requiring $n \geq 3f+1$, which enables the rest of the DAG protocol to choose whether it wishes to work with $n=2f+1$ or more, independently of the equivocation handling mechanism.

\begin{table*}[htbp]
\caption{Summary of our findings regarding safety and liveness (expected termination time) for various $k$ values}
\begin{center}
\begin{tabular}{ c|c|c|c|c } 
\hline
Protocol/safety, liveness  &  $k=2$ &  $k=3$ &  $k>3$ & notes\\  \hline
 DAG-Rider & safe, 2 waves  & safe, 1.5 waves & safe, $\frac{k}{k-1}$ waves  \\ \hline
 Tusk & safe, not live & safe, 3 waves & safe, $\frac{k}{k-2}$ waves & pipelined \\ \hline
 Tusk(Random) & safe, small prob & safe, $\frac{4}{3}$ waves  & safe, 1.06 waves  \\ \hline
 Bullshark Asych & not safe & safe, 1.5 waves  & safe, $\frac{k}{k-1}$ waves  \\ \hline
 Bullshark Partial Synch & safe, live & safe, live & safe, live  \\ \hline
\end{tabular}
\label{tab:Summary}
\end{center}
\end{table*}

\paragraph*{Our Contributions}
\begin{enumerate}
    \item We performed an analysis of DAG-Rider~\cite{Dag-rider}, Tusk~\cite{Tusk}, and Bullshark~\cite{Bullshark} with $n=2f+1$ validators, focusing on safety, liveness, and expected termination where relevant. This study investigates the feasibility of adapting these protocols to operate with $2f+1$ validators, rather than the current $3f+1$ requirement, assuming equivocation can be eliminated. 
    Our findings are summarized in Table~\ref{tab:Summary}. Specifically, whenever the number of validators is $2f+1$ (and equivocation has been eliminated):
    \begin{itemize}
        \item Echoing the findings of~\cite{Letittee}, we show that DAG-Rider can be adapted to work correctly.
        \item We show that while Tusk can maintain safety, it does not ensure liveness.
        \item The asynchronous version of Bullshark does not even guarantee safety.
        \item The partially synchronous version of Bullshark provides both safety and liveness.
    \end{itemize}
    \item We demonstrate that increasing the number of validators beyond $3f+1$ positively impacts expected termination, but with diminishing returns.  
    \item As a minor side contribution, we introduce a new TEE-less equivocation elimination technique, enabling the protocol to construct and order the DAG with the participation of only $2f+1$ validators
    \ifdefined\ICBCCR
    (deferred to~\cite{ourfullpaper}).
    \else
    (deferred to~\appendixref[, section~6]{app:no-equivocation}).
    \fi
\end{enumerate}

\section{Background}
\label{sec:background}

\subsection{Model and Building Blocks}
We assume a typical asynchronous distributed system prone to Byzantine failures.
That is, the system comprises a set of $n$ validators $\{p_1,p_2,...,p_n\}$ where up to $f$ of them may behave arbitrarily, i.e., be \emph{Byzantine}.
The rest of the validators are \emph{honest}, and are assumed to follow the protocol.
We assume reliable links between honest validators, ensuring that every message between them will eventually be delivered, only sent messages may be delivered, and the recipient can verify the sender's identity (no impersonation/Sybil attacks).
DAG-Rider, Tusk, and asynchronous Bullshark protocols assume the communication is completely asynchronous, such that there is no bound on message delay. 
Partially synchronous Bullshark, on the other hand, assumes an asynchronous execution up to an unknown Global Stabilization Time (GST) after which the messages sent between honest validators arrive within a maximal known $\Delta$ delay.

\paragraph*{Reliable Broadcast}
All the protocols we consider use reliable broadcast as a building block.
Reliable broadcast ensures the following properties:

\begin{LaTeXdescription}
\item[Agreement:] If an honest validator delivers a message, then all other honest validators eventually deliver the same message with probability 1.
\item[Integrity:] A message is delivered by each honest validator at most once.
\item[Validity:] If an honest validator broadcasts a message, then all honest validators eventually deliver that message with probability 1.
\end{LaTeXdescription}

\paragraph*{Shared Coin}
Most asynchronous DAG based protocols rely on a shared coin abstraction~\cite{NLS01,CKS00,LJY14,LM18,Sh00}.
Translated into our setting, a shared coin exposes a \textit{id}~:=~\texttt{coin\_toss}($w$) abstraction.
It ensures that all invocations of \texttt{coin\_toss}($w$) return the same \textit{id} value, which is the identifier of one of the validators in the system.
It satisfies the following properties:
\begin{LaTeXdescription}
    \item[Agreement:] If two correct validators call \texttt{coin\_toss}($w$), then the identifiers they return are the same.
    \item[Termination:] If at least $f+1$ validators call \texttt{coin\_toss}($w$), then every call to \texttt{coin\_toss}($w$) eventually returns an identifier.
    \item[Unpredictability:] As long as fewer than $f+1$ validators call \texttt{coin\_toss}($w$), the probability $pr$ that the adversary can guess the return value is $pr \leq 1/n + \epsilon$, for some negligible probablilty $\epsilon$.
    \item[Fairness:] The probability that a call to \texttt{coin\_toss}($w$) would return any given validator's identifier is $1/n$.
\end{LaTeXdescription}


\paragraph*{Byzantine Atomic Broadcast (BAB)}
BAB satisfies all the properties of Reliable Broadcast, along with total order, ensuring they deliver messages in the same order. 
Formally:
\begin{LaTeXdescription}
\item[Total Order:]
If an honest validator delivers a message $m_1$ before another message $m_2$, no other honest validator delivers $m_2$ before $m_1$.  
\end{LaTeXdescription}
Dag-Rider, Tusk, and Bullshark aim to solve the Byzantine Atomic Broadcast (BAB) problem and demonstrate that each of the four properties is satisfied.
DAG-Rider and Bullshark use the notion of weak links to ensure validity, while Tusk does not, and ensures \emph{transaction-level fairness} rather than \emph{block-level fairness}.
We do not mention these mechanisms for addressing validity in this work, as our proposed changes are orthogonal.
Similar to other works, we find it constructive to categorize the protocols using the notions of safety and liveness, which in turn provide all the above properties.
\begin{LaTeXdescription}
    \item[Safety:]
    Total order is preserved.
    \item[Liveness:]
    Progress is ensured, i.e., messages are eventually committed with probability 1.
\end{LaTeXdescription}

\subsection{DAG Construction and Ordering }
The protocols we are interested in operate in a similar manner:
Validators repeatedly reliably broadcast their proposals, and build the next layer (or round) of their local DAGs according to the proposals they have received. 
Each validator then inspects its own view of the DAG and orders it, using no extra communication.

\subsubsection{DAG Construction}
Each vertex in the DAG contains the message reliably broadcast by a certain validator, including references to previously seen vertices by that validator.
These references serve as the edges of the DAG, as illustrated in Figure~\ref{fig:DAG_illustration}.
Different correct validators may see different local DAGs at any given point in time, but reliable broadcast prevents equivocation and guarantees that all correct validators eventually deliver the same vertices. 
We denote $DAG_i$ as the DAG that validator $p_i$ observes.

Each vertex contains a single round number $r$, the source that broadcast it, and a set of $n-f$ edges to vertices belonging to the previous round $r-1$. 
Each validator may observe in its local DAG up to $n$ such vertices for a given round $r$, each associated to a different validator source (including itself).
The reliable broadcast ensures that a validator cannot generate two versions of the same vertex (for the same round).
We denote the set of vertices corresponding to round $r$ in validator $p_i$'s DAG as $DAG_i[r]$.

Whenever there is a sequence of contiguous edges from vertex $u$ to vertex $v$, this is denoted $path(u,v)$. 
The causal history of vertex $v$ in a $DAG_i$ is the set $\{u \in DAG_i| path(v,u)\}$.
When a validator (reliably) delivers a vertex, it adds it into its DAG (assuming the vertex is valid), under the condition that all the vertices pointed to by its edges are already in the DAG. 
Thus, when a vertex is added to the local DAG, it is guaranteed that all its causal history is already in the DAG. 
Note that since the use of reliable broadcast eliminates equivocations, any two validators that add a vertex $v$ broadcast~by validator $p_i$ for round $r$ to their respective DAGs have the exact same $v$, and they also observe the same causal history for~$v$.

As soon as a validator adds $n-f$ vertices of round $r$ to its DAG, it creates and broadcasts its own vertex for round $r+1$, with references to all the vertices of round $r$ it has already seen.
The rule for advancing rounds in Bullshark is a little more subtle.
That is, in Bullshark, there are timeouts that enable waiting for a short additional amount of time, even after receiving $n-f$ messages, before advancing to the next round.
This is done to improve the protocol's termination time in certain favorable scenarios.

Figure \ref{fig:DAG_illustration} illustrates an example for a DAG observed by validator $p_1$, where $f=2$, and the total number of validators is $n=2f+1$.

\begin{figure}[htbp]
\centerline{\includegraphics[width=0.6\linewidth]{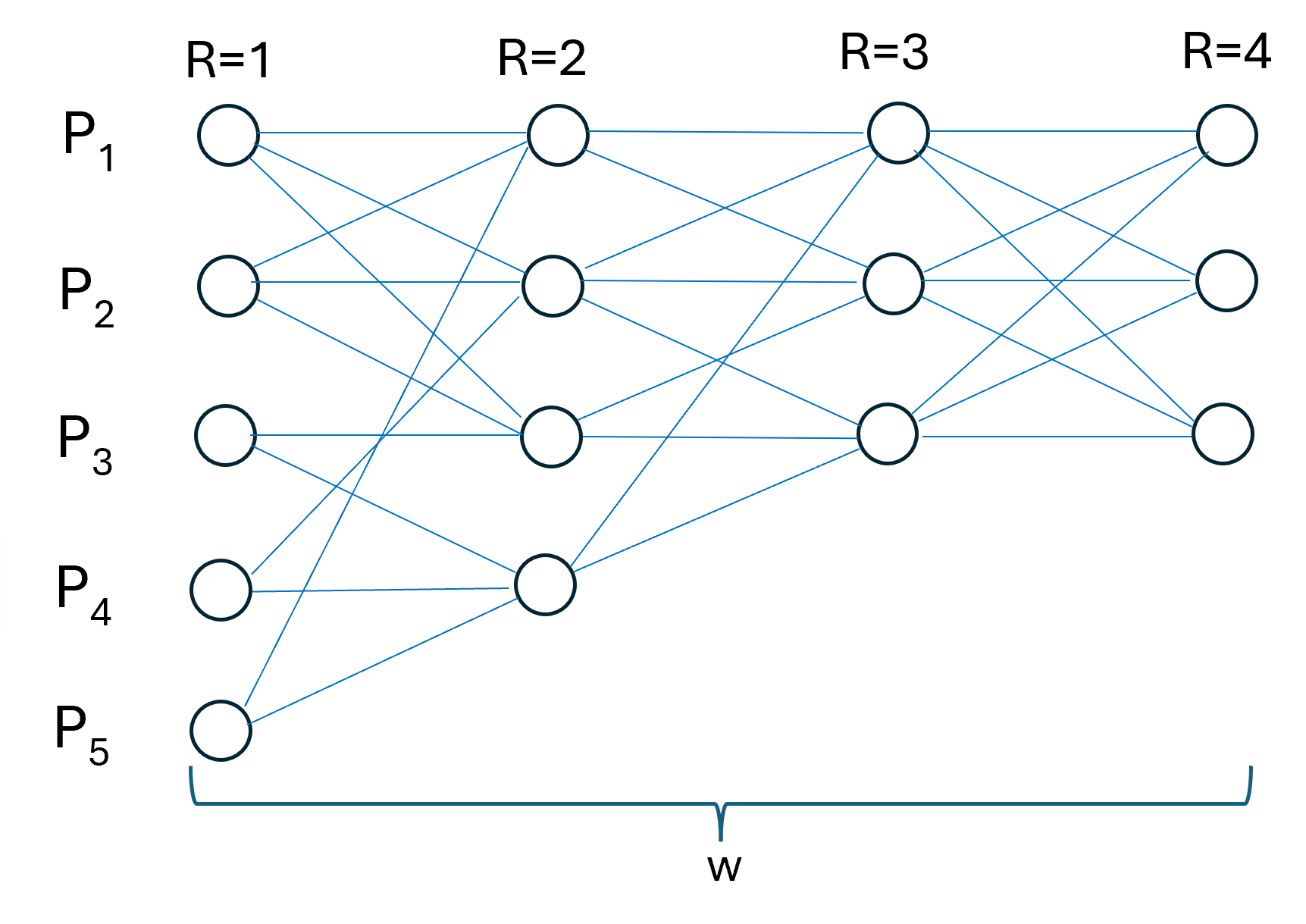}}
\caption{DAG illustration for validator $p_1$. $f=2$, and the total number of validators is $n=2f+1$. 
Each row corresponds to a validator and a column represents the round number.  
Each round contains at least $n-f = 3$ and up to $n=5$ vertices, where each vertex points to at least $n-f = 3$ vertices from the previous round. wave w consists of rounds 1-4.
}
\label{fig:DAG_illustration}
\end{figure}

\subsubsection{Ordering the DAG}


All three protocols partition rounds into waves, where each wave consists of a constant number of rounds. Each protocol aims to finalize a decision at wave boundaries. In Figure \ref{fig:DAG_illustration}, wave 
$w$ consists of four rounds, as is the case with DAG-Rider, for instance.
We annotate the $j^\mathrm{th}$ round of wave $w$ of validator $p_i$ as 
$DAG_i[\textit{round}(w,j)]$.
At least one leader is chosen for the wave, and a commit rule is applied to decide whether a certain leader's proposal can be committed.
The number of rounds that define a wave and the commit rule applied within it are specific to each algorithm.
As established by the famous FLP result \cite{FLP}, Byzantine Atomic Broadcast cannot be deterministically solved in an asynchronous setting. To address this, a global perfect coin is used to introduce randomness: validators choose the leader of the wave using the coin at the end of the wave.
The vertex produced by that leader in round(w,1)
is the candidate to be committed.
We sometimes refer to that vertex as the leader when it is clear from context.
Since the random leader is elected at the end of the wave, liveness is guaranteed: by the time the leader is chosen, it is too late for an adversary to manipulate the network and prevent the commit as the DAGs for the respective wave are already formed.
Bullshark introduces predefined leaders that are chosen deterministically, in addition to the randomly chosen leader. 
Once the leader is committed, the vertices in its causal history are deterministically ordered.
Note that since lack of equivocation is guaranteed, if all validators commit the same leader vertex, they all observe the same causal history, and thereby commit all the history in the same order.
Hence, the main challenge is in preserving total order on the commit of the leaders. 
In particular, since validators may have different views of the DAG at any given time, the protocol should guarantee that if one honest validator commits a leader, then all honest validators will eventually commit it as well.

DAG protocols typically enforce both direct and indirect commit rules, reflecting the principle of quorum intersection. If validator $p_i$ identifies a pattern in wave $w$ and decides to directly commit the corresponding leader, quorum intersection guarantees the presence of another pattern in other DAGs, enabling the other validators to indirectly commit the same leader.
The direct commit rule evaluates whether the leader has received a sufficient number of \emph{votes} within a wave.
Intuitively, a vote signifies that a validator (voter) has observed the leader.
The primary distinction between protocols lies in how they define votes and establish commit rules. For instance, in DAG-Rider, a leader 
$v$ is committed if there are at least 
$2f+1$ vertices in the last round of the wave that have paths to $v$.
Once a leader is directly committed in $DAG_i$, the DAG is traversed backwards through previous rounds/waves to check if any past leaders might have been committed by another validator. An indirect commit rule is then applied to determine this, and if so, the leaders are ordered in ascending order by their rounds, with the earlier leader ordered first.
In DAG-Rider, for example, the indirect commit rule is fulfilled if there is a path from the committed leader to previous leaders.
The safety stems from the fact that fulfilling the  direct commit rule in $DAG_i$ will imply that the indirect rule will be fulfilled in $DAG_j$, for any pair of honest validators $p_i$ and $p_j$.
\subsubsection{Equivocation Elimination}
To enable the aforementioned DAG-based protocols to operate with $n=2f+1$ validators, we assume the existence of a mechanism that prevents equivocation and ensures reliable broadcast. Equivocation can be addressed using a TEE-based approach or, alternatively, a TEE-less method that employs witness validators, whose sole purpose is to ensure reliable broadcast. In both cases, once reliable broadcast is ensured, the rest of the protocol can proceed with 
$n=2f+1$ validators, if desired. 
Both approaches are described in details
\ifdefined\ICBCCR
in~\cite{ourfullpaper}.
\else
in \appendixref[, section~5]{app:no-equivocation}.
\fi

\section {DAG-Rider}
\label{sec:dag-rider}
DAG-Rider was proposed using $3f+1$ validators.
We outline the key aspects of the protocol in a general manner, applicable to a varying number of validators, where the total number of validators is given by 
$n=kf+1$, for any redundancy factor $k\geq2$ with $f$ representing the maximal number of tolerated Byzantine nodes.
Every validator $p_i$ builds and orders its DAG, $DAG_i$, according to the following rules:
\begin{enumerate}
    \item Each round consists of at least $(k-1)f+1$ vertices.
    \item Each vertex points to $(k-1)f+1$ vertices from the previous round.
     \item Each wave is constructed from 4 rounds.
     \item At the end of a wave $w$, a validator $p_i$ is chosen using the shared coin flip abstraction.
     The leader is $p_i$'s vertex in the first round of $w$, and it is the candidate for commit.
     \item Direct commit: if there are $(k-1)f+1$ vertices at round $4$ of $w$ that have paths to the leader $v$, then $v$ is committed:
$|  { v' \in DAG_i[\textit{round}(w,4)] \colon
\textit{path}(v',v) }| \geq
(k-1)f+1
$.
\item Indirect commit: 
    When the leader $v$ of wave $w$ is directly committed, recursively iterate from wave $w-1$ to the last wave for which a leader was committed, and apply the following logic for each wave $w'$, s.t. $w'<w$ :
    if there is a path from $v$ to $v'$ such that $v'$ is an uncommitted leader vertex in a wave $w'$, then $v'$ is also committed.
    The leaders committed this way are ordered in ascending order according to their round numbers such that $v'$ is ordered before $v$.
\end{enumerate}

\subsection{Safety}
Lemma~1 in~\cite{Dag-rider} is the main lemma harnessed to show how total order is satisfied. 
The LetItTee work~\cite{Letittee} augments this lemma to the $2f+1$ case.
We generalize the lemma for any $n=kf+1$ validators, $k\geq2$:
\begin{lemma}
\label{lemma:strongPathLeaders}
If a correct validator $p_i$ commits the wave leader $v$ of a wave $w$, then for any validator $p_j$, any leader vertex $v'$ of a wave $w'>w$ such that $v' \in DAG_j[\textit{round}(w',1)]$ will have a path to $v$.
\end{lemma}

\begin{proof}
Since vertex $v$ is committed by validator $p_i$, there is a set $V$ of $(k-1)f+1$ vertices in $DAG_i[\textit{round}(w,4)]$ that have paths to $v$.
Any vertex $u$, s.t. $u \in DAG_j[\textit{round}(w+1,1)]$ has $(k-1)f+1$ edges pointing to vertices from the previous round. 
Therefore, due to quorum intersection of two subsets of $(k-1)f+1$ out of $kf+1$ possible vertices, and given that no equivocation is possible, $u$ will have a path to at least one (when $k=2$) of the vertices of $V$, and therefore will have a path to $v$.
By induction, any vertex belonging to rounds greater than \textit{round}($w+1$,$1$), including $v'$, will have a path to $v$. 
\end{proof}

\subsection{Liveness}
The key to proving liveness lies in showing that the protocol ensures that the probability for a leader to be committed in any wave is at least $\frac{(k-1)f+1}{kf+1} \approx \frac{k-1}{k}$. 
Hence, the protocol obtains progress in a constant number of waves in expectation.
The main building block used to prove the above probability is called the common core abstraction by Attiya and Welch~\cite{AW04}.
The procedure has 3 all-to-all asynchronous rounds, where in each round a process broadcasts and collects information from $n-f$ processes. 
By the end of the procedure, every process returns the set of inputs it accumulated, and it was proved in~\cite{AW04} that all correct validators have at least $n-f$ common values originating from the inputs of the validators.
DAG-Rider shows that the common core abstraction can be mapped to the DAG construction. 
Lemma~2 of~\cite{Dag-rider}, which was rephrased for the $2f+1$ case by~\cite{Letittee}, is now generalized for any $k\geq2$~below:

\begin{lemma}
\label{lem:commonCore}
Let $p_i$ be a correct validator that completes wave $w$. 
In this case, $\exists V \subseteq DAG_i[\textit{round}(w,1)]$ and $\exists U \subseteq DAG_i[\textit{round}(w,4)]$ s.t. $|V| \geq (k-1)f+1$, $|U| \geq (k-1)f+1$ and $\forall v \in V,\forall u \in U \colon \textit{path}(u,v)$.
\end{lemma}

\begin{proof}
 
We adapt the common-core proof to the DAG in the following way:
Define a table $T$ with $n$ rows and $n$ columns. 
For each $i$ and $j$, entry $T[i,j]$ contains a 1 if the vertex of validator $p_i$ at round $3$ has an edge to the vertex of $p_j$ (at round 2). 
If $p_i$ did not broadcast its vertex in round 3, $T[i,j]$ contains a 1 if and only if $p_j$ did broadcast its vertex in round 2.
If the validator did broadcast its vertex in round 3, the row contains ($n-f$) ones, as a valid vertex of round 3 points to $n-f$ vertices of round 2.
If $p_i$ did not broadcast its vertex in round 3, there are $n-f$ ones, one for each correct validator that broadcast its vertex in round 2.
Each row contains at least $n-f$ ones, and the total number of ones is $n(n-f)$.
Since there are $n$ columns, some column $l$ contains at least $n-f$ ones.
This means that there are less than $f$ vertices that broadcast a vertex in round 3 but do not have an edge to the vertex $v_l$ of $p_l$ in round 2.
Each vertex in the subset of vertices of round 4, $U$, has $n-f$ edges, and as $n-f>f$ for $n>2f$, they will have a path to vertex $v_l$.
Since $v_l$ is valid, it has $n-f$ edges to vertices of round 1, the subset $V$.
Therefore, each vertex in $U$ has a path to the set of vertices $V$ pointed to by $v_l$.
The sizes for $V$ and $U$ are $n-f$, i.e., $(k-1)f+1$.
\end{proof}

Given Lemma~\ref{lem:commonCore}, the chance for a leader of a wave $w$ to belong to the set $V$ (from the lemma statement), is $\frac{(k-1)f+1}{kf+1}$, and at the end of $w$ it will be committed according to the direct commit rule.
Note that since the validators flip the global coin only after they finish wave $w$, the probability that even a powerful adversary would guess the identity of the leader is roughly equal to $1/n$.
Hence, its ability to manipulate the set $V$ is also very limited.
Consequently, the protocol commits a leader in expectation once every $\frac{k}{k-1}$ rounds.
In other words, the tradeoff in choosing $k$ is that a small $k$ value implies better resilience.
Also, smaller DAGs lead to lower memory consumption, lower computational overhead and shorter messages.
Yet, the drawback of a small $k$ is longer expected commit times (higher latency).

\section{Tusk}
\label{sec:tusk}

As before, we augment Tusk~\cite{Tusk} for a varying number of validators, where the total number of validators is $n=kf+1$, with $k\geq2$ being the redundancy factor and $f$ representing the maximum number of Byzantine nodes.
Every validator $p_i$ builds and orders its DAG, $DAG_i$, according to the following:
\begin{enumerate}
    \item Each round consists of at least $(k-1)f+1$ vertices.
    \item Each vertex points to $(k-1)f+1$ vertices from the previous round.
     \item Each wave is constructed from 3 rounds.
     \item Consecutive waves are pipelined such that round 3 of wave $w$ and round 1 of $w+1$ are executed together.
     \item At the end of a wave $w$, a validator $p_i$ is chosen using the shared coin flip abstraction.
     The leader is $p_i$'s vertex in the first round of $w$, and it is the candidate for commit.
     \item Direct commit: 
     if there are $f+1$ vertices of round 2 that have edges to the leader $v$, then $v$ is committed. That is, the condition is met when there are $f+1$ \emph{voters} to the leader:
     $|  { v' \in DAG_i[\textit{round}(w,2)] \colon
\textit{path}(v',v) }| \geq
f+1
$.
    \item Indirect commit: when the leader $v$ of wave $w$ is directly committed, recursively iterate from wave $w-1$ to the last wave for which a leader was previously committed, and apply the following logic for each wave $w'$, s.t. $w'<w$:
    If there is a path from $v$ to $v'$ such that $v'$ is an uncommitted leader vertex in a wave $w'$, then $v'$ is committed as well.
    The leaders committed through this process are ordered in ascending order according to their round~numbers ($v'$ before $v$).
\end{enumerate}

\subsection{Safety}
Lemma 1 in~\cite{Tusk} serves as the key argument for safety.
The lemma statement is unchanged, but we prove that it holds for any $n=kf+1$ validators, $k\geq2$:
\begin{lemma}
\label{lem:tusk_quorum}
If an honest validator $p_i$ commits a leader vertex $v$ in a wave $w$, then any leader vertex $v'$ committed by any honest validator $p_j$ in a future wave will have a path to $v$ in $p_j$'s local DAG.
\end{lemma}

\begin{proof}
$p_i$ commits a (leader) vertex $v$ in a wave $w$ only if there are at least $f+1$ vertices in the second round of $w$ with edges to $v$. 
Since every vertex has at least $(k-1)f+1$ edges to vertices in the previous round, we get by quorum intersection (for every $k\geq 2$, and given that equivocation is impossible) that every vertex in the first round of $w+1$ has a path to $v$. 
By induction, we can show that every vertex in every round in waves higher than $w$, including $v'$, has a path to $v$.
\end{proof}

\subsection{Liveness}
We prove that for $k=2$, liveness is not guaranteed.
To that end, we show a counter example where a powerful adversary controlling the network may prevent the direct commit rule from being met for all validators.
We also prove Tusk's liveness for $k\geq3$.
Specifically, when $k\geq3$, the probability for a leader to be committed in any wave is at least $\frac{(k-2)f+1}{kf+1} \geq \frac{1}{3}$. 
Therefore, the protocol obtains progress in a constant number of waves in expectation.

\paragraph*{For $k=2$, Tusk Does Not Satisfy Liveness!}
We show a counter example of a run in which no validator is able to commit using the above direct commit rule.
Specifically, in Figure~\ref{fig:Tusk_counter}, we depict $DAG_1$ (the top one) and $DAG_2$ (on the bottom) for $k=2$, $f=3$, a total of $n=7$ vertices.
The scenario includes two waves, where round 3 of $w$ and round 1 of wave $w+1$ are pipelined.
By repeating this scenario indefinitely, no commit is possible.
Here, $p_1$ $p_2$, $p_3$ and $p_4$ are honest nodes while $b_1$, $b_2$ and $b_3$ are Byzantine. 
The blue circles represent vertices that were actually (reliably) delivered by $p_1$ while it was executing that round; 4 ($(k-1)f+1$) vertices in each round.
The white vertices are those that were added to the DAG after the the validator has already progressed to the next round. 
For example, in $DAG_1$, round 2 of wave $w+1$, the (white) vertex of $p_2$ was delivered only after $p_1$ has already moved to round~3.
We emphasize that there is a single version of each vertex in all local DAGs, i.e., no equivocation.

Observing $DAG_1$, at the end of round 3 of wave $w$, there are 4 (blue) vertices.
They point to 5 vertices in round 2, which point to 7 vertices in round 1.
However, no vertex in round 1 is pointed to by 4 ($f+1$) vertices from round 2, and thus the commit rule cannot be satisfied.
In wave $w+1$, a very similar situation occurs, resulting in no commit.
The graph for $p_3$ closely resembles the graph for $p_1$, with the vertices of $p_1$ replaced by those of $p_3$. We present the graph for $p_2$ for completeness, showing that no commit is possible there either.
The graph for $p_4$ is symmetric.

\begin{figure}[htbp]
    \centering
    \includegraphics[width=0.65\linewidth]{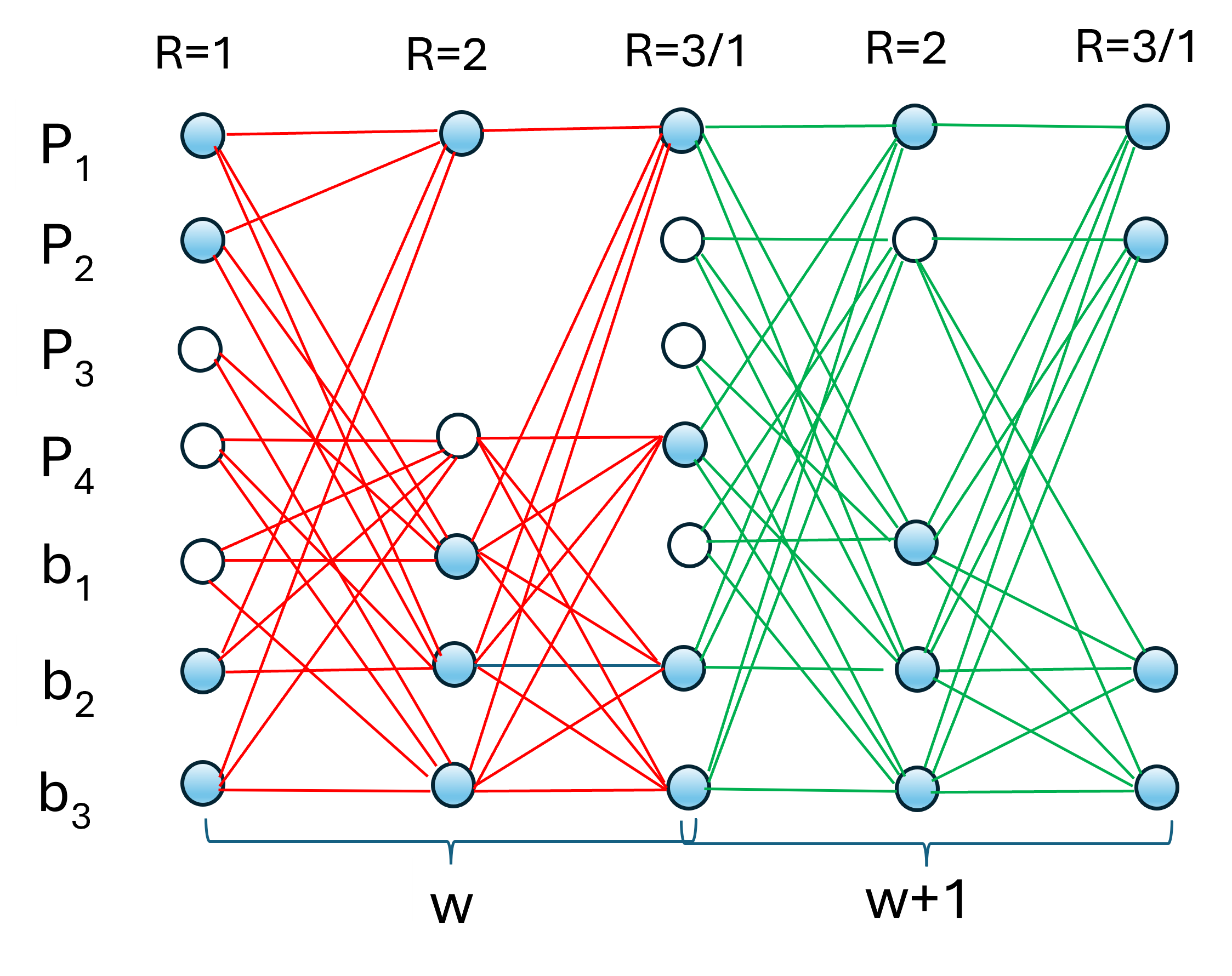}
    \label{fig:DAG1}  
    \includegraphics[width=0.65\linewidth]{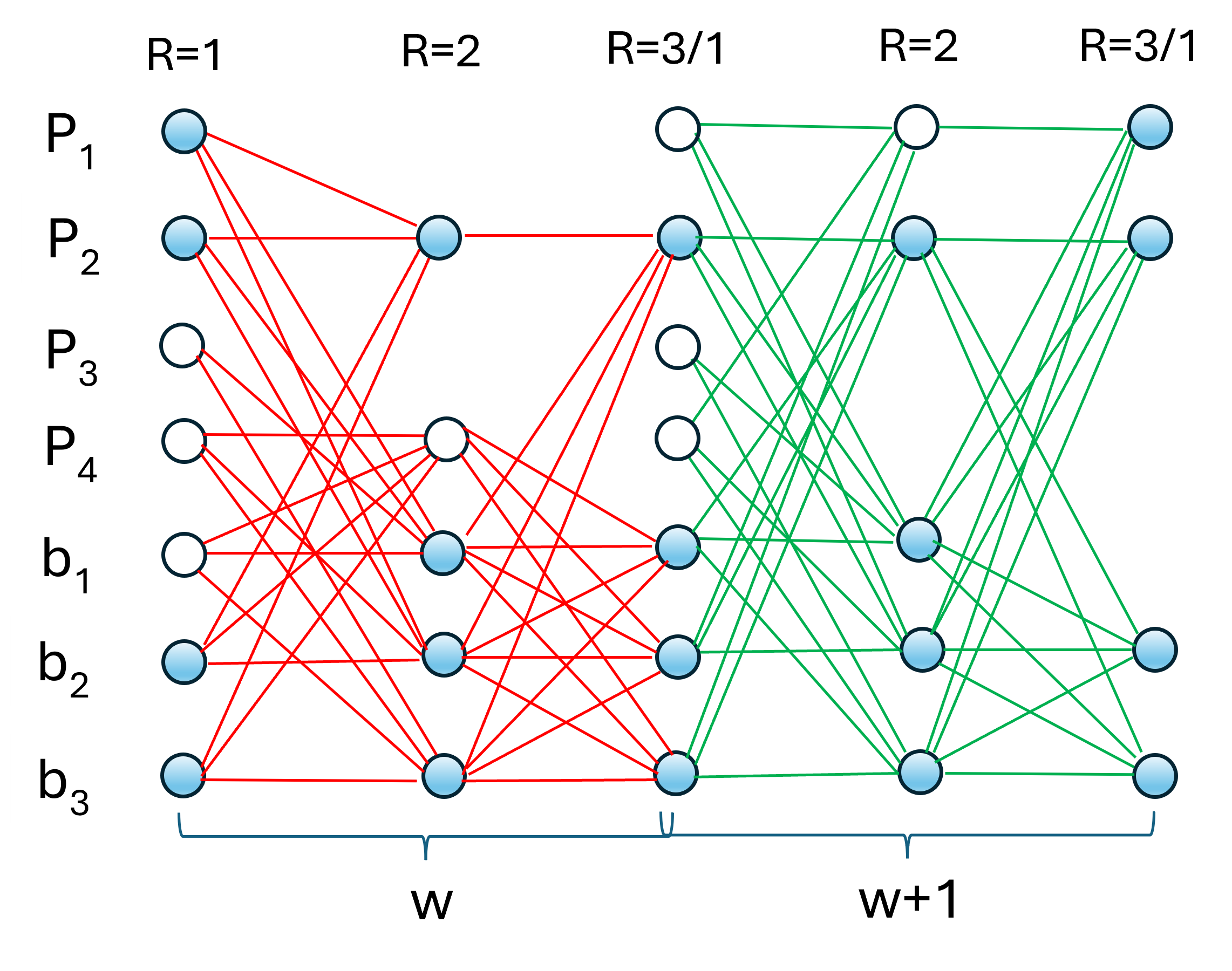}
    \label{fig:DAG2}
    
    \caption{Tusk counter example for $k=2$, $f=3$. Two consecutive waves $w$ and $w+1$ are depicted.
     $DAG_1$ is shown above and $DAG_2$ at the bottom.
     The blue vertices represent the 4 ($f+1$) vertices that are added during a given round. The white vertices are those added only after the validator progressed from that round.
     Any vertex in the first round of $w$ is pointed to by less than 4 ($f+1$). The same happens for $w+1$. Therefore, no commit is possible.}
    \label{fig:Tusk_counter}
\end{figure}

\paragraph*{For $k\geq 3$, Tusk Satisfies Liveness}
We generalize Lemma~3 from~\cite{Tusk} to state the following:
\begin{restatable}{lemma}{tuskliveness}
\label{lemma:tusk-liveness}
For every wave $w$ there are at least $(k-2)f+1$ vertices in the first round of $w$ that satisfy the commit rule, for any $k\geq3$.
\end{restatable}

\ifdefined\ICBCCR
For lack of space, and since this is a generalization of the proof in~\cite{Tusk}, we defer it to~\cite{ourfullpaper}.
\else
Since this is a simple generalization of the proof in~\cite{Tusk}, we defer it to \appendixref[, section~1]{proof:tusk-liveness}.
\fi
Given the above lemma, and the fact that the leader is chosen only at the end of the wave, the probability to elect a leader that satisfies the commit rule in a any given wave $w$ is $\frac{(k-2)f+1}{kf+1} \geq \frac{1}{3}$.

\subsection{Tusk with Random Delays}

The work in~\cite{Tusk} also analyzes the probability of committing a leader when the network operates with random delays rather than being controlled by a powerful adversary.
We show that the commit probability for $k=2$ is $\frac{1}{2}^{f+1}$, for $k=4$ it is $0.94$, and for $k=5$ it becomes~$0.99$.
\ifdefined\ICBCCR
For lack of space, the details appear in~\cite{ourfullpaper}.
\else
The technical details appear in \appendixref[, section~2]{app:tusk-random-delays}.
\fi

\section {Bullshark-Asynchronous Instance}
\label{sec:asynch-bullsharl}
Bullshark~\cite{Bullshark} is a successor of DAG-Rider, which optimizes for synchronous periods. 
To achieve this, Bullshark adds additional leaders to each wave, and introduces two types of votes for vertices: \emph{steady-state} for the predefined leader and \emph{fallback} for the random~one.

Each wave in Bullshark consists of 4 rounds, where the first round has two potential leaders, a \textit{steady-state} and a \textit{fallback} leader. 
The third round of a wave also has a predefined leader, the second steady-state leader.
Being predefined, the steady-state leaders are chosen deterministically, while the fallback leader is selected at the end of the wave (in round~4) based on the shared random coin.
Although there are two types of leaders in each wave, both fallback and steady-state leaders should never be committed within the same wave.

Each validator has a voting type for a certain wave $w$, and its vertices can vote in that wave accordingly. 
Hence, they cannot vote for both the fallback and steady-state leaders in the same wave $w$. 

Validator $p_i$ determines $p_j$'s voting type in wave $w$ once $p_j$'s vertex $v$ belonging to the first round of $w$ is added to $DAG_i$.
Once $p_i$ adds $p_j$'s vertex $v$, it can observe all its causal history for wave $w-1$ (Byzantine validators cannot falsify this) and determine according to the following:
If $p_j$ managed to commit either the fallback leader or the second steady state leader in wave $w-1$, its type for wave $w$ is determined as steady-state;
otherwise, as fallback.
Note that this means that $p_j$'s voting type in wave $w$ is consistent across all validators that receive $v$, as they all observe the same $v$ (no equivocation), hence the same causal history. 

Vertices in the second round of a wave $w$, round($w$, $2$), can vote to the first steady state leader, assuming they have steady-state voting type.
Vertices of the $4^\mathrm{th}$ round, round($w$, $4$), can vote for the second  steady state leader, which is in round($w$, $3$), or for the fallback leader, which is in round($w$, $1$), depending on their voting type.
Voting is ``done'' simply by having a path to the leader.

Since Bullshark's goal is to take advantage of synchronous periods, it includes timeouts in the DAG construction.
We follow Bullshark's construction and commit rule and extend them as follows:
\begin{enumerate}
    \item Each round consists of at least $(k-1)f+1$ vertices.
    \item Each vertex points to at least $(k-1)f+1$ vertices from the previous round.
    \item Advancing to even rounds: Advance to the second and fourth rounds of a wave only if (1) the timeout for this round expired or (2) the wave's predefined first or second steady-state leader, respectively, has been delivered.
    \item Advancing to odd rounds: Advance to the third round of a wave or to the first round of the next wave if (1) the timeout for this round expired or (2) $(k-1)f+1$ vertices in the current round with steady-state voting type and edges to the first and second steady-state leader, respectively, have been delivered.
    \item Upon delivery of $p_i$'s vertex in round $1$ of wave $w$, consider the vertices $v$ points to as ``potential votes''.
    These vertices belong to wave $w-1$ and each of them has a voting type that was already previously determined (when they were delivered in rounds of wave $w-1$).
    Using these votes, try to commit either the second steady state leader or the fallback leader of wave $w-1$ (leaders with different types should never be committed within the same wave) as follows:
    \begin{enumerate}
        \item Direct commit, fallback: try to commit the fallback leader of wave $w-1$ according to the direct commit rule: at least $(k-1)f+1$ out of the potential votes must have paths to the leader and a fallback voting~type. 
        \item Direct commit, $2^\mathrm{th}$ steady-state: try to commit the second steady-state leader of wave $w-1$ according to the direct commit rule: at least $(k-1)f+1$ out of the potential votes must have paths to the second steady-state leader and a steady-state voting~type. 
        \item Determine votes: if either of these leaders is committed, determine $p_i$'s voting type for wave $w$ as steady-state. 
        Otherwise, determine it as fallback.
    \end{enumerate}
    \item Direct commit, first steady-state: Upon delivery of $p_i$'s vertex $v$ in round($w$,$3$), consider the vertices~$v$ points to as ``potential votes''.
     Try to commit the first steady-state leader of wave~$w$ according to the direct commit rule: at least $(k-1)f+1$ out of the potential votes must have paths to the first steady-state leader and a steady-state voting type.
    \item Indirect commit: when a leader $v$ (steady-state or fallback) is directly committed, traverse backwards and try to commit previous leaders which were not committed yet, i.e., candidate leaders, as follows:
    \begin{enumerate}
    \item 
    For a steady-state candidate leader in round $r$, potential votes = \{$u \in  DAG[r+1] ~|~ path(v,u)$\}.
    \item For a fallback candidate leader in round $r$, 
    potential votes = \{$u \in DAG[r+3] ~|~ path(v,u)$\}.
    \item Compute actual votes: the potential voters that have a path to the candidate leader; Set the actual votes for fallback to the empty set in rounds where there is no fallback leader or when a steady-state leader has already been committed in the wave.
    \item If one of the leaders has at least $(k-2)f+1$ votes, while the other type has at most $f$, order the leader (in ascending order according to its round) and continue traversing~backwards.\label{indirect_decision}
    \end{enumerate}
\end{enumerate}

\subsection{Asynchronous Bullshark Safety}
We show that the Asynchronous Bullshark augmentation for $k=2$ does not maintain safety, but it is safe for any $k\geq 3$.

\paragraph*{Asynchronous Bullshark is Not Safe with $k=2$}

    
For intuition, we first briefly explain why safety requires $k\geq3$. 
Safety in Bullshark relies on the fact that direct commit of a leader's vertex requires $(k-1)f+1$ validators of the same type as the leader (steady-state/fallback) with paths to the leader's vertex.
This implies that any other correct validator sees at least $(k-2)f+1$ of these votes in its DAG, and $f$ or less validators with the other type (which serves as the indirect commit rule). 
Note that if $(k-2)f + 1$ are steady-state (fallback) voters, a fallback (steady-state) leader could not have been committed directly, as there are at most $2f$ such voters, which is insufficient to form a direct commit in the case that $k \geq 3$. A direct commit requires at least $(k-1)f + 1 \geq 2f$~voters.

To see why safety breaks when $k=2$, suppose a validator $p_i$ directly committed a steady-state leader due to detecting $f+1$ steady-state voters.
Another validator $p_j$ may see only one of these voters in $DAG_j$, which is not enough to determine whether there could have been $f+1$ steady-state voters or $f+1$ fallback voters.
The $f$ potential voters, which $p_j$ does not know about, could be either steady-state voters or fallback~voters.

We demonstrate in Figure~\ref{fig:BS_counter} a safety violation of the specified protocol for $k=2$. We depict in the example the DAG for a wave $w$ of validator $p_4$ when there are $2f+1$ validators in total ($k=2, f=2$). 
At the beginning of $w$, $p_4$ determines the voting types for the vertices of $p_1$, $p_4$ and $p_5$. 
Suppose $p_1$ is determined as steady state (colored orange in the figure), while $p_4$ and $p_5$ are fallback voters (green vertices) for wave $w$.
Assume that $p_4$'s vertex of round $1$ was elected (using the shared coin) as the fallback leader at the end of the wave, while the (first) steady-state leader of $w$ is the vertex of $p_5$.
Assume also that $p_4$ managed to directly commit a leader in wave $w'>w$.
According to the protocol, when a leader is directly committed, the validator traverses the DAG backwards to try and indirectly commit leaders that may have been committed by other validators.

When $p_4$ reaches wave $w$, and tries to decide whether to indirectly commit a leader in round $1$ of $w$, it has two candidates. 
At this point, $p_4$ computes the votes for each potential leader:
it finds 2 fallback voters and 1 steady-state voter for each of the leaders.
$p_4$ needs to choose according to the indirect commit rule specified in Item~\ref{indirect_decision} above.
However, both conditions are fulfilled:
the steady-state leader has at least $(k-2)f+1 = 1$ steady-state voters and at most $f = 2$ fallback voters, and the fallback leader has at least $(k-2)f+1 = 1$ fallback voters and at most $f = 2$ steady-state voters.
Since the (indirect) rule is fulfilled, $p_4$ may commit a leader in $w$. 
However, any choice made by $p_4$ could be wrong, as the leader of the other type might have been committed (depending on the voting types of $p_2$ and $p_3$).

\begin{figure}[htbp]
\centering
\includegraphics[width=0.65\linewidth]{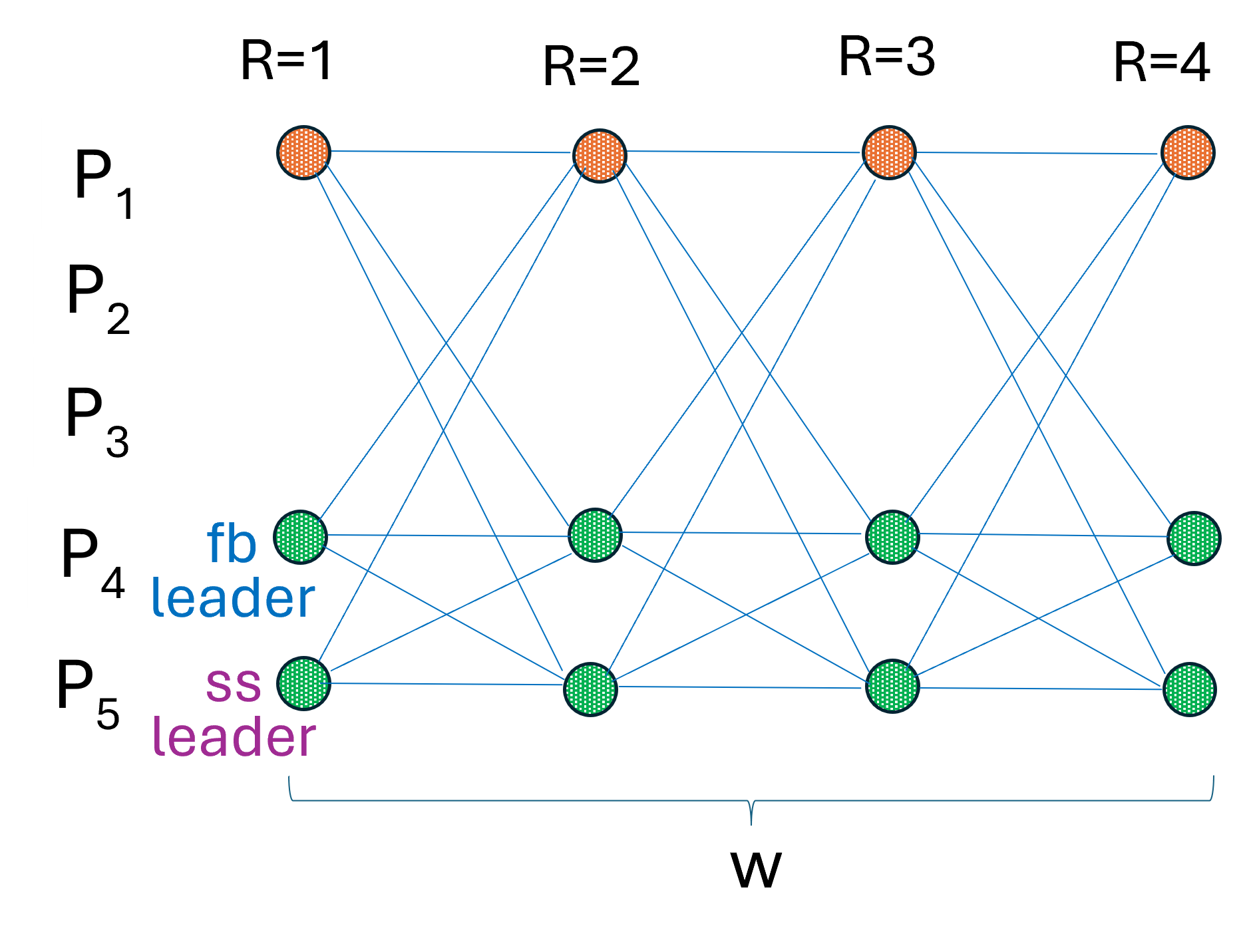}
\caption{\label{fig:BS_counter} Asynchronous Bullshark showing safety violation for $k=2$. Consider the DAG of validator $P_4$ at wave $w$, $f=2$.
The steady-state vertices of $p_1$ are colored orange, while the green vertices of $p_4$ and $p_5$ have fallback voting type.
Validator $p_4$ has the indirect rule fulfilled for both types of leaders in $w$. 
Any leader it chooses to commit could be the wrong one.}
\end{figure}

\paragraph*{Asynchronous Bullshark is Safe when $k\geq 3$}

\ifdefined\ICBCCR
For lack of space, and since the proofs are merely a generalization of the ones in~\cite{Bullshark}, we provide the details in~\cite{ourfullpaper}.
\else
Since the proofs are merely a generalization of the ones in~\cite{Bullshark}, we provide the details in \appendixref[, section~3]{app:asynch-bullshark-safe}.
\fi

\subsection{Asynchronous Bullshark Liveness (when \texorpdfstring{$k\geq 3$}{k >= 3})}
\label{sec:asynch-bullshark-live}

Since we showed that when $k=2$, Bullshark is not safe, for the liveness part we focus on the generalized case that $k\geq 3$, and show it to be live.
\ifdefined\ICBCCR
As before, due to lack of space and since the proofs are a generalization of the ones in~\cite{Bullshark}, we state them in~\cite{ourfullpaper}.
\else
As before, due to lack of space and since the proofs are a generalization of the ones in~\cite{Bullshark}, we state them in \appendixref[, section~4]{app:asynch-bullshark-live}.
\fi

\section {Bullshark - Partially Synchronous Instance}
\label{sec:synch-bullshark}

The partial synchrony version of Bullshark~\cite{Bullshark} has no fallback leaders, only two predefined (steady-state) leaders for each wave. 
The protocol works as follows:

\begin{enumerate}
    \item Each round consists of at least $(k-1)f+1$ vertices.
    \item Each vertex points to at least $(k-1)f+1$ vertices from the previous round.
    \item Advancing to even rounds: A validator advances to the second and fourth rounds of a wave only if (1) the timeout for this round expired or (2) it delivered the wave's predefined first or second leader, respectively.
    \item Advancing to odd rounds: A validator advances to the third round of a wave or the first round of the next wave, if (1) the timeout for this round expired or (2) it delivered $(k-1)f+1$ vertices in the current round with edges to the first or second leader, respectively.\label{advanceES_BS}
   \item Direct commit~1: when a vertex $v$ is added in the first round of wave $w$, consider the vetices $v$ points to as votes. If $f+1$ of these votes have a path to the second leader (which is in the third round) of wave $w-1$, commit this~leader.
    \item Direct commit~2: when a vertex $v$ is added in the third round of wave $w$, consider the vertices $v$ points to as votes. If $f+1$ of these votes have a path to the first leader (which is in the first round) of wave $w$, commit this leader.
    \item Indirect commit: 
    when a leader $v$ is directly committed, recursively traverse backwards and try to commit uncommitted leaders. If there is a path from $v$ to the candidate leader $v'$, commit $v'$.
    The leaders are ordered in ascending order according to their round numbers, such that $v'$ is ordered before $v$.
\end{enumerate}
\subsection{Safety}
We claim that the partial synchrony instance is safe for $k\geq2$.
The following lemma establishes total order on the ordering of the leaders, and is proven
\ifdefined\ICBCCR
in~\cite{ourfullpaper}
\else
in~\appendixref[, section~5]{app:ps-bullshark}:
\fi
\begin{restatable}{lemma}{psbullsharksafety}
\label{lemma:ps-bullshark-safety}
If an honest validator $p_i$ commits a leader vertex $v$, then any leader vertex $v'$ committed by any honest validator $p_j$ in a future round will have a path to $v$ in $p_j$'s local DAG. 
\end{restatable}

\subsection{Liveness}
The statement of Claim~9 in~\cite{Bullshark} indicates that after GST, if there are two consecutive honest leaders, then the second one will be committed by all honest validators.
The proof does not rely on the size of $k$.
Rather, the proof relies on the properties of reliable broadcast and the rules for advancing rounds.
Therefore, it can be adapted almost verbatim.

Given that the timeouts are greater than $3\Delta$, and  both leaders of a wave $w$ are honest, they show all honest validators advance to round($w$,$4$) within $\Delta$ of each other, and they all deliver and add the second leader in round($w$,$3$) to their DAGs.
Thus, at the beginning of round $4$, when broadcasting its vertex, each such honest validator adds an edge to the second leader vertex.
Then, before advancing to the first round of the next wave, and since they are at most $\Delta$ time away from each other, the timeout is large enough so that they wait for each other's vertices (described above in bullet~\ref{advanceES_BS}).

Thus, each honest validator will get $(k-1)f+1$ vertices in round($w$,$4$) with paths to the second leader of $w$, and will commit this leader according to the direct commit rule.
It should be mentioned that the underlying assumption stating that eventually there will be a wave with two contiguous honest leaders is true for every $k\geq2$.

\section{Related Work}
\label{sec:related}

\paragraph*{Consensus Termination in Non-DAG Algorithms}
The impact of larger quorums on the complexity and termination time of classical consensus protocols has been studied in the context of classical algorithms, e.g., termination in one communication step in favorable runs~\cite{FMR05}.
A protocol that reaches asynchronous Byzantine consensus in two communication
steps in the common case with $5f+1$ validators was presented in~\cite{Fab}.
However, the protocol of~\cite{Fab} only ensures liveness during periods of synchrony.
Recently, an improved tight bound of $5f-1$ validators was given in~\cite{KTZ21}.
Some works, such as~\cite{ezBFT,Zyzzyva}, have also considered a combination of a fast path requiring quorums of $3f+1$ (out of $3f+1$) nodes with slow paths that work with $2f+1$ sized quorums.
Many works have also addressed the issue of improving the expected termination time of asynchronous protocols while ensuring maximal resiliency (of $3f+1$), e.g.,~\cite{AMS19,Aleph,MR17}.


\paragraph*{DAG Based Algorithms}
Shoal~\cite{Shoal}, Shoal++~\cite{Shoal++} and Sailfish~\cite{Sailfish} are recent DAG based protocols that assume a partially synchronous model.
Shoal~\cite{Shoal} enhances the partially synchronous Bullshark with pipelining and a zero overhead leader reputation mechanism. 
Since the basic structure follows Bullshark, we conjecture that it can be adapted to utilize only $2f+1$ validators in a similar manner to our Bullshark augmentation.
Shoal's successor, Shoal++\cite{Shoal++}, attempts to further reduce end-to-end latency, by employing key ideas such as treating all vertices as leaders and operating multiple DAGs in parallel.
Sailfish~\cite{Sailfish} supports a leader vertex in each round. 
In addition, it facilitates multiple leaders within a single round.
\ifdefined\ICBCCR
Studying the augmentation of these protocols to $2f+1$ validators is left for future work.
\else
Studying whether these protocols can be augmented to work with $2f+1$ validators when equivocation is eliminated, e.g., through the use of a TEE, is left for future work.
\fi

CordialMiners~\cite{CordialMiners} and Mysticeti~\cite{mysticeti} are DAG protocols that do not eliminate equivocation prior to inserting a proposal into the DAG.
\ifdefined\ICBCCR
\else
Instead, they address equivocation as part of the algorithm itself, thus circumventing the costs related to reliable broadcast.
\fi
Exploring the impact of incorporating a TEE mechanism into these algorithms is beyond the current scope and is left for future work.

\section{Discussion}
\label{sec:discussion}
Table \ref{tab:Summary} summarizes our findings. 
Based on our analysis, when there are $2f+1$ validators participating in the protocol, DAG-Rider is proven to maintain both safety and liveness. 
Tusk preserves safety.
However, since its liveness relies on a counting argument that requires a sufficient number of votes in the round following the leader's proposal, it fails under the $2f+1$ assumption.
Asynchronous Bullshark, while similar to DAG-Rider, has two types of leaders and therefore requires larger quorums to differentiate between them, resulting in compromised safety with $2f+1$ validators.
Partially synchronous Bullshark preserves safety and relies on the Global Stabilization Time (GST) and two consecutive honest leaders to ensure liveness, which is achievable in the $2f+1$ case.
\ifdefined\ICBCCR
\else
Coming up with general design guidelines for a protocol ensuring both safety and livenesss with $2f+1$ validators remains an open problem which is left for future work.
\fi

Using more than $3f+1$ validators presents a trade-off: while it reduces the expected termination time for asynchronous protocols, the returns are diminishing fast. 
It also leads to larger graphs, higher communication overhead, and lower~resilience.


\clearpage

\bibliographystyle{plain}
\bibliography{refs}

\ifdefined\ICBCCR
\else
\ifappendix
\newpage

\appendices
\section{Tusk liveness}\label{app:sec:tusk-liveness} \label{proof:tusk-liveness}
\tuskliveness*
\begin{proof}
We follow the proof line presented in~\cite{Tusk}:
Consider any set $S$ of $(k-1)f+1$ vertices in the second round of wave $w$.
The total number of edges they have to the first round is
$\left((k-1)f+1\right)\left((k-1)f+1\right) = (k-1)^2f^2+2(k-1)f+1$.
The number of possible vertices in the first round of $w$ is $kf+1$.
Therefore, even if every vertex in the first round has $f$ voters (edges pointing to it) from
vertices in $S$, there are still $(k-1)^2f^2+2(k-1)f+1 - f(kf+1) = (k^2-3k+1)f^2+(2k-3)f+1$ edges.
The maximum number of edges from vertices in $S$ to each vertex in the first round is $(k-1)f+1$.
This is the maximum number of voters from $S$ any vertex may have.
After subtracting the $f$ voters we already accounted for, we get that each round 1 vertex can contain no more than $(k-1)f+1-f=(k-2)f+1$ votes.
Therefore, at least $\frac{(k^2-3k+1)f^2+(2k-3)f+1}{(k-2)f+1} \geq (k-2)f+1$ round 1 vertices will have $f+1$ votes from vertices in $S$.
Note that the above inequality holds only when $k\geq3$.
\end{proof}

\section{Tusk with Random Delays}
\label{app:tusk-random-delays}

We now analyze the probability that Task will commit a leader in a given round when the network operates under random network delay for $k=2$, $k=4$ and $k=5$.
Let $S$ be the set of vertices of round($w$,$2$) of a wave $w$.
Message delays are distributed uniformly at random, and each vertex in $S$ points to $(k-1)f+1$ vertices of round$(w,1)$ independently of other vertices in $S$. Thus, the probability for a vertex in $S$ to point to the leader is $\frac{(k-1)f+1}{kf+1}$.
According to Tusk's commit rule, the leader is committed when $f+1$ such vertices from $S$ point to (vote for) the leader.

When $k=2$, $\frac{(k-1)f+1}{kf+1} = \frac{1}{2}$, there are $f+1$ vertices in S, resulting in the probability to commit being $\frac{1}{2}^{f+1}$.
The $k = 3$ case was analyzed in \cite{Tusk} where it was proved that the probability is at least $0.74$.
The computation was done for $f=1$, because for bigger $f$ the probability is~higher.

We repeat the computation for $k=4$:
Each vertex in $S$ points to the leader with
probability of at least $\frac{(4-1)f+1}{4f+1} \geq 3/4$.
Given that there are 
$3f+1=4$ vertices in $S$, the probability that at least $f+1=2$ out of these $4$ vertices point to the leader is calculated as follows:
$1-\left((\frac{1}{4})^4 + (\frac{3}{4})^1*(\frac{1}{4})^3*4\right)=243/256=0.94$.

To complete the picture, for $k=5$,
the probability that at least $2$ out of the 5($4f+1$) vertices in $S$ to point to the leader:
$1-\left((\frac{1}{5})^5 + (\frac{4}{5})^1*(\frac{1}{5})^4*5\right) = 3104/3125 = 0.99$.

\section{Asynchronous Bullshark is Safe when \texorpdfstring{$k\geq 3$}{k >= 3}}
\label{app:asynch-bullshark-safe}

The following claims closely mirror claims 2, 3, and 4 in~\cite{Bullshark}, but generalized for $k\geq 3$.

\begin{claim}\label{claim:one_type_commit}
    If a validator $p_i$ directly commits a steady-state leader in wave $w$, no other honest validator commits (directly or indirectly) a fallback leader in $w$, and vice versa. 
\end{claim}
    
\begin{proof}
If $p_i$ directly committed a steady-state leader in wave $w$, it determined $(k-1)f+1$ voters as having a steady-state type in $w$.

    Given that there are $kf+1$ validators in total, and that a validator's voting type within a wave remains consistent for all validators (since it is determined by its causal history upon entering $w$), any other validator $p_j$ will identify at most $f$ validators as fallback voters in $w$. Therefore, no fallback leader can be committed by any validator in wave $w$.
\end{proof}

For the next claim, we say that a validator $p_i$ consecutively directly commits two leaders $v_i$ of round $r_i$ and $v_i'$ of round $r_i' > r_i$, if it directly commits them, but does not directly commit any leader between $r_i$ and $r_i'$ .

\begin{claim}\label{claim:commit_indirectly}
    Assume that an honest validator $p_i$ consecutively directly committed $v_i$ and $v_i'$ in rounds $r_i$ and $r_i'$ respectively.
    Suppose an honest validator $p_j$ committed a leader $v_j$ of round $r_j$ such that $r_i \leq r_j \leq r_i'$. 
    In this case, $p_i$ will indirectly commit $v_j$.
\end{claim}
\begin{proof}
    Let $v_j'$ in round $r_j'$ be the ``last'' leader committed by $p_j$ between $r_i$ and $r_i'$, i.e., $r_j'$ is the highest round of a committed leader by $p_j$, s.t. $r_i \leq r_j' \leq r_i'$.
    If we show that $p_i$ indirectly commits $v_j'$, it then follows that $p_i$ also indirectly commits $v_j$ ($v_j'$ has the same causal history in both $DAG_i$ and $DAG_j$), thus completing the proof.
    
    We will show that $v_j'$ will be indirectly committed by $p_i$ when traversing backwards from the commit of $v_i'$.     
    Consider two cases:
    \begin{enumerate}    
        \item $v_j'$ was directly committed by $p_j$.
        Let $r$ be the smallest round between $r_j'$ and $r_i'$ in which $p_i$ committed (directly or indirectly) a leader $v$.
        Examine two scenarios:
        \begin{itemize}
            \item $v_j'$ is a steady-state leader. 
            $ r> r_j'+1$ since leaders are considered in odd rounds only. Since $p_j$ directly committed $v_j'$, there is a set of 
            $(k-1)f+1$ vertices in $DAG_j[r_j'+1]$ with paths to $v_j'$ and with steady-state type.
            Since vertices' types (in a wave) are determined consistently for $p_i$ and $p_j$, and by quorum intersection, there are at least $(k-2)f+1$ vertices in $DAG_i[r_j'+1]$ with steady state type and paths from $v$ to them.
            \item $v_j'$ is a fallback leader. By claim \ref{claim:one_type_commit},
            no leader is committed in round $r_j'+2$, thus $r> r_j'+3$. 
            Since $p_j$ directly committed $v_j'$, there is a set of 
            $(k-1)f+1$ vertices in $DAG_j[r_j'+3]$ with paths to $v_j'$ and with fallback type.
            By quorum intersection, and since vertices' types (in a wave) are determined consistently for $p_i$ and $p_j$, there are at least $(k-2)f+1$ vertices in $DAG_i[r_j'+3]$ with $v_j'$ vote type (fallback) and paths from $v$ to them.
        \end{itemize}
        In both scenarios, $p_i$ will count at least $(k-2)f+1$ votes for the leader $v_j'$, and less than $f$ for the opposite type. Given that  $k\geq3$, we have $(k-2)f+1 > f$, leading to an unambiguous decision.
        Therefore $p_i$ indirectly commits~$p_j'$.
        \label{safety_consecutive}
        \item 
        $v_j'$ was indirectly committed by $p_j$. Consider the two consecutive direct commits of $p_j$ enclosing round $r_j'$. With the roles reversed, and according to the previous case, $p_j$ will indirectly commit $v_i'$, resulting in $v_j' == v_i'$, thereby completing the proof.
    \end{enumerate}    \end{proof}

We establish a total order by applying claim~\ref{claim:commit_indirectly} for any two~validators.

\section{Asynchronous Bullshark Liveness with \texorpdfstring{$k\geq 3$}{k >= 3}}
\label{app:asynch-bullshark-live}
We show that for every round $r$, there is a round $r'>r$, such that an honest validator will commit a leader in $r'$ with probability 1.
We follow claims 5, 6, 7 and 8 from~\cite{Bullshark}, and generalize them for~$k\geq3$:

\begin{claim}
\label{claim:allfb}
Consider a wave $w$ such that for any following 
wave no honest validator has committed a leader.
Then in every wave $w'>w+1$ every honest
validator $p_i$ will determine the voting type of all validators in $DAG_i$ as fallback.
\end{claim}

\begin{proof}
By the claim's assumption, there is no commit starting from wave $w+1$.
When $p_i$ adds a vertex $v$ of round$(w',1)$ by a validator $p_j$ to $DAG_i$, it determines $p_j$'s voting type for $w'$.
It examines whether $p_j$ committed in wave $w'-1$.
According to the assumption $p_j$ did not commit any leader in $w'-1$, so its voting type is set as fallback.
Furthermore, all validators that add round$(w',1)$ vertex from $p_j$ will add the same vertex $v$ (and thus will see the same causal history) and determine $p_j$'s type as~fallback.
\end{proof}

Next, we use Lemma~2 that specifies the common-core property, and the fact that fallback leaders are chosen in~retrospect.

\begin{claim}
\label{claim:probcommit}
If a validator $p_i$ determines that the type of all validators in $DAG_i$ for a wave $w$ are fallback, then the probability of $p_i$ to commit the fallback vertex leader of $w$ is at least~$\frac{k-1}{k} > \frac{2}{3}$ for $k\geq3$.
\end{claim}

\begin{proof}

By the assumption, the vote type of all validators with vertices in $DAG_i[\textit{round}(w,4)]$ is fallback.
By Lemma~2, there exist two sets, 
$U$ and $V$, each of size $(k-1)f+1$, such that, $U\subseteq DAG_i[\textit{round}(w,4)]$ and
$ V \subseteq DAG_i[\textit{round}(w,1)]$,
where each vertex in $U$ has a path to each vertex in $V$.
If any of the vertices in $V$ is chosen to be the leader, $p_i$ will directly commit it. 
Since the fallback leader is elected only in round $4$, it is too late for the adversary to control who is pointing to the leader.
Thus, the probability for the elected leader to be in $V$ is at least $\frac{(k-1)f+1}{kf+1} > \frac{2}{3}$ for $k\geq3$.
\end{proof}

\begin{claim}
For every wave $w$, there is an honest validator that commits a leader in a wave higher than $w$ with probability 1.
\end{claim}

\begin{proof}
Assume that for all waves higher than $w$, no leader was committed. By claim \ref{claim:allfb}, a validator $p_i$ will set the type of all validators of all waves $w'>w+1$ to be fallback.
According to claim~\ref{claim:probcommit}, $p_i$ has a probability of $\frac{k-1}{k}$ to commit the fallback leader in all such waves.
Hence, there is a wave higher than $w$ in which $p_i$ will commit with probability~1. 
\end{proof}

\section{partially synchronous Bullshark safety with \texorpdfstring{$k\geq2$}{k >=2}}
\label{app:ps-bullshark}

\psbullsharksafety*
\begin{proof}
$p_i$ commits a (leader) vertex $v$ in round $r$ if there are at least $f+1$ vertices in round $r+1$ with edges to $v$. 
Since every vertex has at least $(k-1)f+1$ edges to vertices in the previous round, we get by quorum intersection (for every $k\geq 2$, and given that equivocation is impossible) that every vertex in round $r+2$ has a path to $v$. 
Therefore, by induction, we can show that every vertex in every round higher than $r+2$, including $v'$ (leaders are considered in odd rounds only), has a path to $v$.   
\end{proof}

\section{Equivocation Elimination}
\label{app:no-equivocation}

\subsection{On the Use of Trusted Execution Environments}

A na\"ive use of trusted execution environments (TEE) simply runs the entire code of a crash fault tolerant consensus based protocol inside the TEE to obtain a Byzantine fault-tolerant protocol with $n=2f+1$~\cite{VSCBLV13}.
Alas, most existing TEEs execute code much slower than the typical speed outside the TEE.
More importantly, long and complex code is likely to include bugs and vulnerabilities.
Hence, executing an entire consensus based protocol inside the TEE runs a higher risk of being compromised, in which case a single corrupted node could bring down the entire crash resilient protocol.

To that end, multiple TEE assisted BFT protocols that focus on minimal use of the TEE and only to the degree required to eliminate equivocation have been developed~\cite{HybridsSteroids,A2M,DAMYSUS,CheapBFT,Trinc,YANR22}.
Most of these protocols work with $n=2f+1$.
Yet, recently there is a debate in the community about the pros and cons of requiring $n=3f+1$ even when equivocation is eliminated through judicious use of TEEs~\cite{vivisectingdissection,dissecting-bft}.

\subsection{TEE-less Equivocation Removal}

As a minor side contribution, we now show a simple TEE-less augmentation of $3f+1$ reliable Byzantine broadcast that eliminates equivocation, while enabling the rest of the DAG protocol to utilize only $n=2f+1$ validators, whenever the DAG protocol works correctly with $2n+1$ validators once equivocation is eliminated.
As shown before, DAG-Rider is an example of such a protocol.

Specifically, assume we have a total of $3f+1$ nodes.
We divide them into $2f+1$ validators plus $f$ witness nodes.
Only validators participate in the higher level consensus/blockchain/SMR DAG-based protocol, and only they may issue messages to the reliable Byzantine broadcast protocol.
All $3f+1$ nodes participate in the reliable Byzantine broadcast protocol, but only validators deliver messages to their respective DAG protocol layer.

\begin{figure}[htbp]
    \centering
    \includegraphics[width=0.7\linewidth]{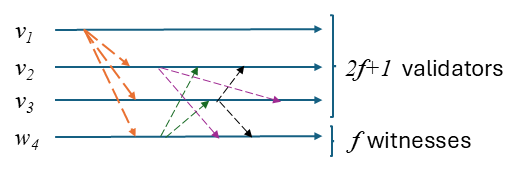}
    \caption{A Simple Signed TEE-less Reliable Broadcast Protocol}
    \label{fig:teeless}
\end{figure}

Figure~\ref{fig:teeless} illustrates such a setting for the simple echo based signed messages reliable broadcast protocol.
The sender first sends a signed copy of its message to all other nodes in the system.
Every node that receives a signed and validated message for the first time, re-sends this message to all other nodes, except the one it came from.
Once a validator receives the same signed and validated message $m$ from $n-f$ nodes, it locally delivers $m$.    

\fi
\fi
\end{document}